\documentclass{article}

\usepackage{microtype}
\usepackage{amsmath,amsfonts,amsthm,amssymb,mathtools}
\usepackage{graphicx,color}
\usepackage{cite}

\usepackage{fullpage}

\title{Almost All Even Yao-Yao Graphs Are Spanners}

\author{Jian Li \qquad\qquad\qquad Wei Zhan
 \\
 \textrm{lijian83@mail.tsinghua.edu.cn}, \quad zhanw-13@mails.tsinghua.edu.cn,
}
\theoremstyle{definition}
\newtheorem{defi}{Definition}

\theoremstyle{plain}
\newtheorem{thm}{Theorem}

\newtheorem{col}[thm]{Corollary}
\newtheorem{lm}[thm]{Lemma}
\newtheorem{cj}[thm]{Conjecture}

\newcommand{\ov}{\overrightarrow}
\newcommand{\pset}{\mathcal{P}}

\newcommand{\trapezoid}[1]{T_#1}
\newcommand{\ttheta}{\trapezoid{\theta}}
\newcommand{\tone}{\Gamma_1}
\newcommand{\ttwo}{\Gamma_2}

\newcommand{\R}{\mathbb{R}}

\newcommand{\YY}[1]{\mathsf{YY}_{#1}}
\newcommand{\TY}[1]{\mathsf{TY}_{#1}}
\newcommand{\OY}[1]{\mathsf{OY}_{#1}}
\newcommand{\Yao}{\mathsf{Y}}

\newcommand{\rotation}[2]{{#1}^{\circlearrowleft#2}}

\begin{document}
  \maketitle

  \begin{abstract}
    It is an open problem whether Yao-Yao graphs $\YY{k}$ (also known as sparse-Yao graphs)
    are all spanners when the integer parameter $k$ is large enough.
    In this paper we show that, for any integer $k\geq 42$, the Yao-Yao graph $\YY{2k}$ is a $t_k$-spanner, with stretch factor $t_k=6.03+O(k^{-1})$ when $k$ tends to infinity.
    Our result generalizes the best known result which asserts that
    all $\YY{6k}$ are spanners for $k\geq 6$ [Bauer and Damian, SODA'13].
    Our proof is also somewhat simpler.
  \end{abstract}

\section{Introduction}
  Let $\pset$ be a set of points on the Euclidean plane $\R^2$.
  A simple undirected graph $G=(V,E)$ with vertex set $V=\pset$ is called a \textit{geometric $t$-spanner},
  if for any pair of vertices $(u,v)$ there is a path $uu_1\ldots u_s v$ in $G$ such that the total length of the path
  $|uu_1|+|u_1u_2|+\ldots+|u_sv|$ is at most $t$ times the Euclidean distance $|uv|$, where $t$ is a constant,
  called the \textit{stretch factor}.
  The concept of geometric spanners was first proposed in \cite{chew1986there},
  and can be considered as a special case of general graph spanners~\cite{peleg1989graph} if we consider the complete graph with edge lengths defined by the Euclidean metric.

  {\em Yao graph} was first introduced by Andrew Yao in his seminal work on high-dimensional Euclidean minimum spanning trees~\cite{yao1982constructing}.
  Let $C_u(\gamma_1, \gamma_2)$ be the cone with apex $u$ and consisting of the rays with polar angles in $[\gamma_1,\gamma_2)$.
  Given a positive integer parameter $k$, the construction of Yao graph $\Yao_k$ is
  described in the following process.
  \begin{center}
    \begin{minipage}[h]{\linewidth}
    \begin{tabbing}
      \hspace*{0.8cm}\=\hspace*{0.8cm}\= \kill
      Initially $\Yao_k$ is an empty graph. \\
      For each point $u$: \\
      \> For each $j=0,\ldots,k-1$: \\
      \> \> Let $C=C_u(2j\pi/k,2(j+1)\pi/k)$; \\
      \> \> Select $v\in C\cap \pset$ such that $|uv|$ is the shortest; \\
      \> \> Add edge $\ov{uv}$ into $\Yao_k$.
    \end{tabbing}
    \end{minipage}
  \end{center}
  In the above process, ties are broken in an arbitrary but consistent manner.	
  The above process is usually referred to as a ``Yao step''.
  Note that the edges $\ov{uv}$ we added are directed edges (the directions are useful
  in the construction of Yao-Yao graphs).

  Besides the usefulness in constructing minimum spanning trees,
  Yao graphs are also sparse graphs with surprisingly nice spanning properties.
  It is known that Yao graph $\Yao_k$ has a stretch factor $(1-2\sin(\pi/k))^{-1}$ when $k\geq 6$
  (see e.g., \cite{li2002sparse} and \cite{bose2012pi}).
  Later, the spanning properties of $\Yao_4$, $\Yao_6$ and $\Yao_5$ are proved in a series of work
  \cite{damian2009spanner,damian2010yao,barba2014new}.
  Due to the spanning properties and the simplicity of the construction,
  they have important applications in networking and wireless communication;
  we refer to \cite{narasimhan2007geometric} and \cite{li2008wireless} for more details.

  One may notice that a Yao graph may not have a bounded degree.
  This is a serious drawback in certain wireless networking applications, since a wireless node
  can only communicate with a bounded number of neighbors.
  The issue was initially realized by Li et al.~\cite{li2002sparse}.
  To address the issue, they proposed a modified construction with two Yao steps:
  the first Yao step produces a Yao graph $\Yao_k$,
  and the second step, called the ``reverse Yao step'',
  eliminates a subset of edges of $\Yao_k$ to ensure the maximum degree is bounded.
  The reverse Yao step can be described by the following procedure:

  \begin{center}
    \begin{minipage}[h]{\linewidth}
    \begin{tabbing}
      \hspace*{0.8cm}\=\hspace*{0.8cm}\= \kill
      Initially $\YY{k}$ is an empty graph. \\
      For each point $u$: \\
      \> For each $j=0,\ldots,k-1$: \\
      \> \> $C=C_u(2j\pi/k,2(j+1)\pi/k)$; \\
      \> \> Select $v\in C\cap \pset$, $\ov{vu}\in \Yao_k$ such that $|uv|$ is the shortest; \\
      \> \> Add edge $\ov{vu}$ into $\YY{k}$.
    \end{tabbing}
    \end{minipage}
  \end{center}

  The resulting graph, $\YY{k}$, is named as ``Yao-Yao graph'' or ``Sparse-Yao graph''
  in the literature. The node degrees in $\YY{k}$ are clearly upper-bounded by $2k$.
  It has long been conjectured that $\YY{k}$ are also geometric spanners
  when $k$ is larger than some constant threshold \cite{li2002sparse,li2008wireless,bauer2013infinite}:
  \begin{cj}\label{cj}
    There exists a constant $k_0$ such that for any integer $k>k_0$, $\YY{k}$ is a geometric spanner.
  \end{cj}

   In sharp contrast to Yao graphs, our knowledge about the spanning properties of Yao-Yao graphs
   is still quite limited.
   Li et al.~\cite{li2002sparse} proved that $\YY{k}$ is connected for $k>6$, and provided extensive experimental evidence
   suggesting $\YY{k}$ are indeed spanners for larger $k$.
   Jia et al.~\cite{jia2003local} and Damian \cite{damian2008simple} showed that $\YY{k}$ is spanner in certain special cases (the underlying point set satisfies certain restrictions).
   It is also known that for small constant $k$'s, none of $\YY{k}$ with $2\leq k \leq 6$ admits a constant stretch factor \cite{barba2014new,damian2009spanner,el2009yao}.
   Recently, a substantial progress was made by Bauer and Damian~\cite{bauer2013infinite}, who showed
   that $\YY{6k}$ are spanners with stretch factor $11.76$ for all integer $k\geq6$ and the factor drops to $4.75$ when $k\geq 8$.
   In fact, a closer examination of the proofs in their paper
   actually implies an asymptotic stretch factor of $2+O(k^{-1})$.
   None of $\YY{k}$ with other $k$ values have been proved or disproved to be spanners.
   We note here that some of the aforementioned work \cite{li2002sparse,damian2008simple}
   focused on \textit{UDGs (unit disk graphs)}.
   But their arguments can be easily translated to general planar point sets as well.
   \footnote{Suppose a spanning property holds for any UDG.
   	For a general set of points, we first scale it so that its diameter is less than 1. We can see that the UDG defined over the scaled set is a complete graph. So the spanning property also holds for the set of points.}

  In this paper, we improve the current knowledge of Yao-Yao graphs and make a step towards
  the resolution of Conjecture~\ref{cj},
  by showing that almost all Yao-Yao graphs with even $k$ are geometric spanners for $k$ large enough.
  Formally, our results is as follows:

  \begin{thm}\label{thm:main}
    For any $k\geq 42$, $\YY{2k}$ is a $t_k$-spanner, where $t_k=6.03+O(k^{-1})$.
  \end{thm}

  \subparagraph{Our Technique}
  Our proof contains two major steps.
  \begin{enumerate}
  \item (Section~\ref{s:oy})
  We first introduce two classes of intermediate graphs, called
  {\em overlapping Yao graphs}
   (denoted as $\OY{k}$)
  and {\em trapezoidal Yao graphs}
   (denoted as $\TY{k}$).
   The construction of overlapping Yao graphs
   is similar to that of Yao graphs, except that
   the cones incident on a vertex overlap with each other.
  We can easily show that $\OY{k}$ is a geometric spanner.
  The definition of trapezoidal Yao graphs is also similar to Yao graphs, but takes advantage of a shape called {\em curved trapezoid} (defined in Section~\ref{s:pre}).
  We can show that $\OY{k}$ is a subgraph of $\TY{k}$, implying
  $\TY{k}$ is a geometric spanner as well.
  \item
  (Section~\ref{s:main})
  In the second step, we show that $\YY{2k}$ spans $\TY{2k}$
  (i.e., for any $u,v \in \pset$,
  the shortest $u$-$v$ path in $\YY{2k}$ is
  at most a constant times longer than the shortest $u$-$v$ path in $\TY{2k}$).
  \end{enumerate}
  Our proof makes crucial use of the properties of curved trapezoids.
  Roughly speaking, curved trapezoids are more flexible than triangles which were used in \cite{bauer2013infinite}),
  which is the main reason for the improvement from $\YY{6k}$
  to more general $\YY{2k}$.

  \subparagraph{Related Works}
  Replacing the Euclidean distance $|\cdot|$ by a power $|\cdot|^\kappa$ with a constant $\kappa\geq 2$
  leads to the definition of \textit{power spanners}.
  Since the power of length models the energy consumed in wireless transmissions,
  power spanners have important implications in wireless networking applications.
  In this setting, Yao-Yao graphs $\YY{k}$ have been proved to be power spanners for any $\kappa\geq 2$ when $k>6$~\cite{jia2003local,schindelhauer2007geometric}.
  It is clear that when a graph is a geometric spanner, it must also be a power spanner (See \cite[Lemma 1]{li2001power}),
  however the reverse does not hold.

  We also note that neither Yao graphs nor Yao-Yao graphs can be guaranteed to be planar graphs
  \cite{schindelhauer2007geometric,kanj2012certain},
  whereas Delaunay triangulation provides another type of spanner which is planar but without bounded degree
  \cite{dobkin1987delaunay,bose2010stretch}.
  In order to facilitate network design in certain applications, some previous work   \cite{bose2002constructing,li2003efficient}
  made use of both
  Yao graphs and Delaunay graphs to produce degree-bounded and planar geometric spanners.

  The paper is organized as follows. We introduce some standard notations and useful tools,
  including the shape of curved trapezoid, in Section \ref{s:pre}.
  We introduce overlapping Yao graphs, trapezoidal Yao graphs, and prove their spanning properties
  in Section \ref{s:oy}. We prove
  our main result in Section \ref{s:main}.
  Finally, we conclude with some future work in Section \ref{s:con}.

\section{Preliminaries}\label{s:pre}

    $\pset$ is the underlying set of points in $\R^2$.
    $D(a,\rho)$ denotes an open disk centered at point $a$ with radius $\rho$.
    The boundary and closure of a region $R$ are denoted by $\partial R$ and $\overline{R}$, respectively.
    Let $S(a,\rho)=\partial D(a,\rho)$ be the circle centered at $a$ with radius $\rho$.
    The length of shortest $u$-$v$ path in a graph $G$ is denoted by $d_G(u,v)$.

  On $\R^2$, a point $u$ with coordinates $(x,y)$ is denoted by $u(x,y)$.
  Let $o(0,0)$ be the origin of $\R^2$.
  The positive direction of $x$-axis is fixed as the polar axis throughout
  the construction and analysis.
  For a point $a\in \R^2$, we use $x_a$ to denote its $x$-coordinate and $y_a$ its $y$-coordinate.
  We use $\varphi(uv)$ to represent the polar angle of vector $\ov{uv}$.
  The angle computations are all under the modulo of $2\pi$, and angle subtraction is regarded as the counterclockwise difference.
  A cone between polar angles $\gamma_1<\gamma_2$ with apex at the origin is
  denoted as $C(\gamma_1,\gamma_2)=\{u\mid \varphi(ou)\in[\gamma_1,\gamma_2)\}$.

  It is also necessary to introduce some notations of standard affine transformations on a geometric object $\mathcal{O}$ in the plane:
  \begin{itemize}
    \item[-] (Dilation) If $\mathcal{O}$ is uniformly scaled by factor $\lambda$ with the origin as the center, the result is denoted by $\lambda\mathcal{O}=\{\lambda z\mid z\in\mathcal{O}\}$.
    \item[-] (Translation) If $\mathcal{O}$ is translated so that the original point goes to point $u$, the result is denoted by $u+\mathcal{O}=\{u+z\mid z\in\mathcal{O}\}$.
    \item[-] (Rotation) If $\mathcal{O}$ rotated an angle $\gamma$ counterclockwise with respect to the origin, the result is denoted by   $\rotation{\mathcal{O}}{\gamma}$.
    \item[-] (Reflection) If $\mathcal{O}$ is reflected through the $x$-axis, the result is denoted by $\mathcal{O}^-$.
  \end{itemize}

  By the above notations, we can denote the cone with apex $u$
  by $u+C(\gamma_1,\gamma_2)$ (abbreviated as $C_u(\gamma_1,\gamma_2)$).
  \begin{figure}[t]
    \centering
  {
    \includegraphics[width=0.45\linewidth]{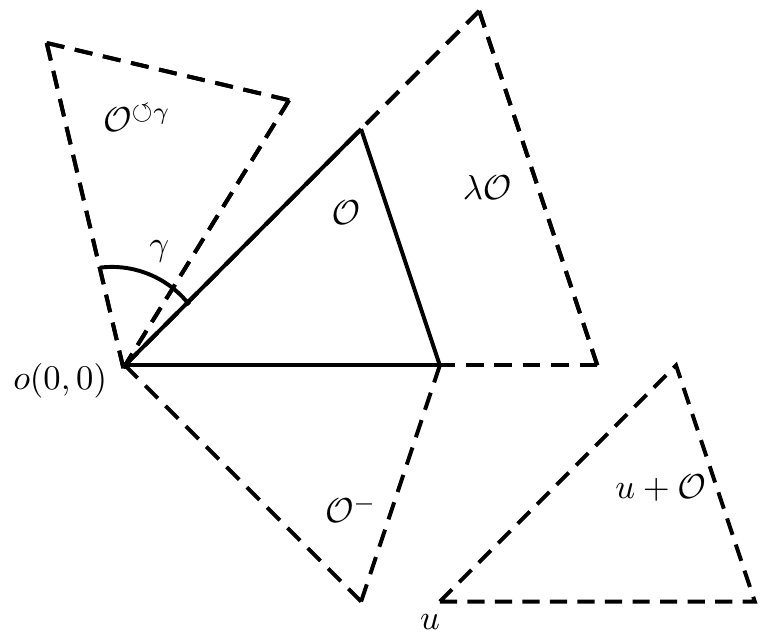}
  }\quad
  {
    \includegraphics[width=0.35\linewidth]{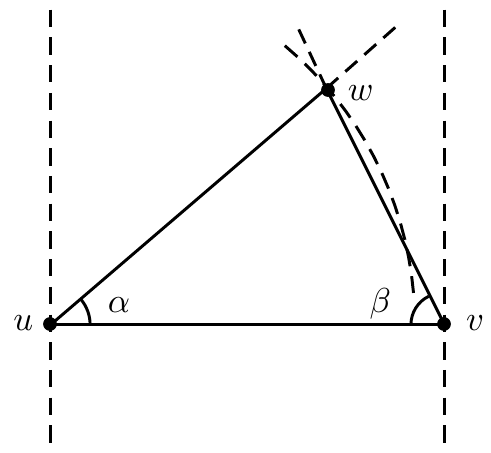}
  }
  \caption{Left: Affine transformations. Right: Proof for Lemma \ref{lm:ineq}.
  }
  \label{fig:ineq}
  \end{figure}
  \paragraph*{Geometric Inequalities}
  In order to attain a constant stretch factor for Yao graphs and their variants,
  one often needs to bound a certain geometric ratio.
  Here we present a simple yet general lemma for this purpose. Let the {\em open strip} on the segment $uv$ be the collection of points
  \begin{lm}\label{lm:ineq}
    Suppose $\tau\geq 1$ is a constant.
    Let $w$ be a point with $\angle wuv,\angle wvu\in[0,\pi/2)$ (see Figure \ref{fig:ineq}).
    When $\tau|vw|<|uv|$, the following statements about the ratio $|uw|/(|uv|-\tau|vw|)$ hold:
    \begin{enumerate}
		\item If $w$ is restricted within a compact segment of an arc centered at $u$, the ratio is maximized when $|vw|$ is the largest;
      \item If $w$ is restricted within a compact segment of a ray from $v$, the ratio is maximized when $|vw|$ is the largest;
     \item If $w$ is restricted within a compact segment of
      a ray originated from $u$,
      the maximum ratio is achieved
      when $w$ makes $|uw|$ the largest or smallest;
      \end{enumerate}
  \end{lm}
  \begin{proof}
    The first statement
    is straightforward since both $|uv|$ and $|uw|$ are fixed.
    Now, we show the last two properties.
    Let $\alpha=\angle wuv$ and $\beta=\angle wvu$.
    By the law of sines and the fact that
    $\sin(\alpha+\beta)=\sin\alpha\cos\beta+\cos\alpha\sin\beta$, we can see that
    \begin{equation}\label{eq:lm}
    \dfrac{|uw|}{|uv|-\tau|vw|}=\frac{\sin\beta}{\sin(\alpha+\beta)-\tau\sin(\alpha)}
    =\frac{1}{\cos\alpha-\sin\alpha\cdot(\tau-\cos\beta)/\sin\beta}.
    \end{equation}
    One can see the ratio is maximized when $\alpha$ is maximized, which implies the
    second statement.
    Moreover, $(\tau-\cos\beta)/\sin\beta$ is a convex function of $\beta$ for $\beta\in (0,\pi/2)$.
    Hence, it is maximized when $\beta$ is maximized or minimized,
    which implies the last statement.
  \end{proof}
   Due to Lemma \ref{lm:ineq}, to compute an upper bound for
   $|uw|/(|uv|-|vw|)$ in a region,
   it suffices to consider a small number of extreme positions.
   In particular, the following corollary is an immediately consequence,
   which is an improvement of a classical result mentioned in \cite{czumaj2004fault}:
  \begin{col}\label{col:ineq}
    $R=\overline{D(u,|uv|)\cap C_u(\gamma_1,\gamma_2)}$ is a sector with apex  $u$ and with $v$ on its arc. Suppose $\alpha=\max\{\varphi(uv)-\gamma_1,\gamma_2-\varphi(uv)\}<\pi/3$.
    Then, for all $w\in R\setminus\{u\}$,
    it holds that
    $$
    \dfrac{|uw|}{|uv|-|vw|}\leq\left(1-2\sin\frac{\alpha}{2}\right)^{-1}.
    $$
  \end{col}
  \begin{proof}
    Take $\tau=1$ in Lemma \ref{lm:ineq}. By considering the rays from $u$ and the arcs centered at $u$, and applying the first and third cases in Lemma \ref{lm:ineq}, we can see that the ratio can only be maximized when $w$ is at the two corners of the sector with $\beta=(\pi-\alpha)/2$, or when $w$ approaches $u$ with $\beta=0$. The two upper bounds in equality~\ref{eq:lm} in these cases are $\left(1-2\sin(\alpha/2)\right)^{-1}$ and $(\cos\alpha)^{-1}$, and the former one is larger.
  \end{proof}

  \paragraph*{Curved Trapezoid $\ttheta$}

  To construct the trapezoidal Yao Graphs $\TY{k}$,
  we need to define an open shape, called \textit{curved trapezoid}, as follows.

  \begin{defi}[Curved Trapezoid $\ttheta$]\label{def:ct}
  Fix two points $o(0,0)$ and $p(1,0)$.
  Given $\theta\in[\pi/4,\pi/3)$, we define the curved trapezoid $\ttheta$ as:
  \[\ttheta = \{u(x,y)\mid 0< x < 1, 0< y< \sin\theta, |ou|<1, |pu|< 1\}.\]
  \end{defi}
  Intuitively, it is the convex hull of two sectors with apices at $o$ and $p$ respectively.
  We regard $\ttheta$ as a shape attached to the origin $o$,
  and the closed arc not incident on $o$ is called the \textit{critical arc}.
  See Figure \ref{fig:t} for an example.
  \begin{figure}[t]
    \centering
  {
    \includegraphics[width=0.3\linewidth]{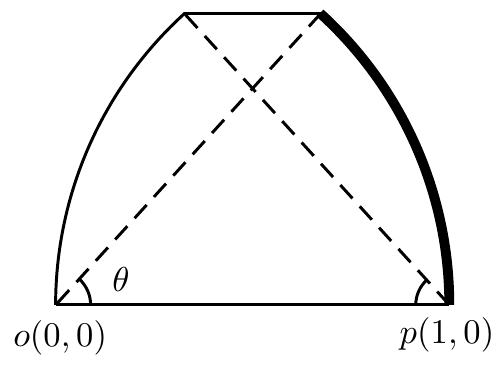}
  }\qquad
  {
    \includegraphics[width=0.5\linewidth]{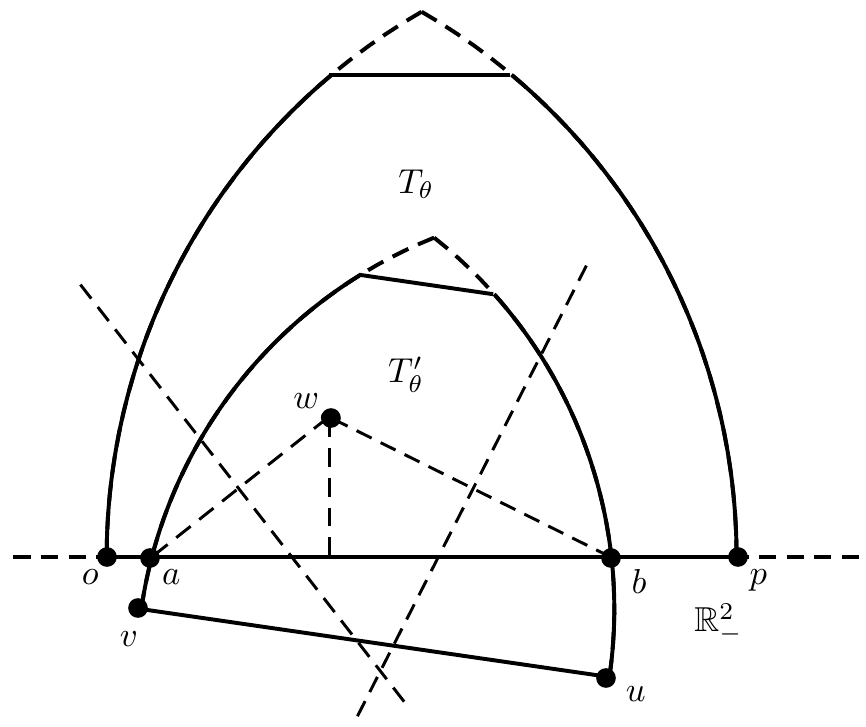}
  }
  \caption{Left: The curved trapezoid. The critical arc is shown in bold.
  Right: An illustration of Lemma \ref{lm:lmt} with its proof.}
  \label{fig:t}
  \end{figure}

  We list some useful properties of $\ttheta$ here.
  It is straightforward to show that $\ttheta$ is symmetric with respect to the vertical line $x=\frac{1}{2}$.
  For any $u\in \ttheta$ it holds that $0<\varphi(ou)<\pi/2$.
  If $u$ is on the critical arc, it holds that $0<\varphi(ou)<\theta$.
  The following crucial lemma, illustrated in Figure \ref{fig:t},
  is also an immediate consequence from the definition.
  \begin{lm}\label{lm:lmt}
    Denote $\mathbb{R}^2_-$ as the lower half-plane $\{(x,y)\mid y\leq 0\}$.
    Consider any two points
    $u,v\in\mathbb{R}^2_-$
    satisfying $0<x_u<1$, $|\varphi(uv)-\pi|<\dfrac{1}{6}\pi$, and $|ou|,|pv|\in[|uv|,1)$.
    We construct a similar shape $\ttheta':=u+|uv|\rotation{(\ttheta^-)}{\varphi(uv)}$.
    We have that $\ttheta'-\ttheta$ is completely contained in lower half plane $\mathbb{R}^2_-$.
  \end{lm}

  \begin{proof}
  We first consider the extreme case when $\theta=\pi/3$. Since $u,v\in\mathbb{R}^2_-$, if the two arcs of $\ttheta'$ do not intersect with the line $y=0$, then $\ttheta$ is fully contained in $\mathbb{R}^2_-$ and the proof is done.

  Now suppose the two intersections are $a,b$, illustrated in the RHS of Figure \ref{fig:t}. Notice that $|ua|=|uv|\leq|ou|$
  and $|vb|=|uv|\leq|pv|$. Combining with the condition that $0<x_u<1$, we can see that both $a$ and $b$ lie in the segment $op$. For any point $w\in \ttheta'-\mathbb{R}^2_-$, it suffices to prove $|ow|,|pw|<1$ so that $w$ is in $\ttheta$. This is done by examining the perpendicular bisector of $aw$ and $bw$: since $|ua|>|uw|$, $u$ is in the upper half-plane of the bisector of $aw$, and
  so is $p$. Thus $|pw|<|pa|\leq |op|=1$. The same arguments hold for $|ow|<1$.

  The above shows that $\ttheta'-\ttheta\subset\mathbb{R}^2_-$ when $\theta=\pi/3$, and if $\theta$ is smaller, we only need to show that the distance $d(w,op)$ from $w$ to line $y=0$ is less than $\sin\theta$. This is done by the observation that $d(w,ab)\leq d(w,uv)$ in $\ttheta'$, and $d(w,uv)<|uv|\sin\theta<\sin\theta$. This completes the proof for Lemma \ref{lm:lmt}.
\end{proof}

\section{Overlapping Yao Graphs and Trapezoidal Yao Graphs}\label{s:oy}
  In this section we consider two variants of Yao graphs,
  overlapping Yao graphs and trapezoidal Yao graphs.

  \begin{defi}[Overlapping Yao Graph $\OY{k}$]\label{def:oy}
    Let $\gamma=\left\lceil\dfrac{k}{4}\right\rceil\dfrac{2\pi}{k}$.
    For every $u\in\pset$ and $j=0,\ldots,k-1$, select shortest $\ov{uv}$ with $v\in C_u(2j\pi/k,2j\pi/k+\gamma)$.
    The chosen edges form the \textit{overlapping Yao graph} $\OY{k}(\pset)$.
  \end{defi}

  The angle $\gamma$ in the definition is actually the smallest multiple of $2\pi/k$
  which is no less than $\pi/2$.
  We note here that the term ``overlapping"
  comes from the fact that the cones with the same apex overlap with each other,
  while in the original Yao graphs they are disjoint.
  First we claim that $\OY{k}$ is a spanner when $k$ is large.

  \begin{lm}\label{lm:oy}
    If $k>24$, then $\OY{k}$ is a $\tau_k$-spanner where $\tau_k=\left(1-2\sin\big(\dfrac{\pi}{k}+\dfrac{\pi}{8}\big)\right)^{-1}$.
  \end{lm}
  \begin{figure}[t]
    \centering
    \includegraphics[width=0.4\linewidth]{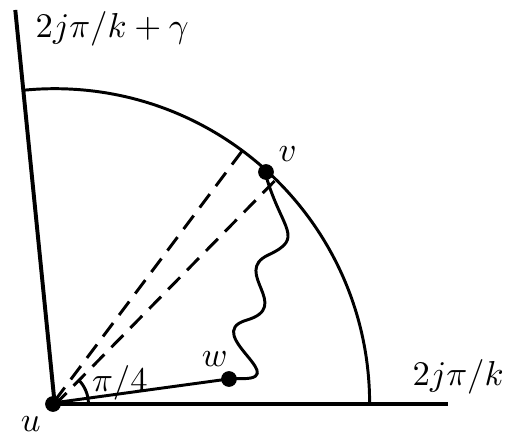}
    \caption{Illustration of the proof of Lemma \ref{lm:oy}.
		The outer cone is $C_u(\gamma_1,\gamma_2)$
		where $\gamma_1=2j\pi/k$ and $\gamma_2=2j\pi/k+\gamma$
    	The cone defined by two dashed lines is $C_u\left(2j\pi/k+\pi/4,2(j+1)\pi/k+\pi/4\right)$.}
    \label{fig:oy}
  \end{figure}
  \begin{proof}
    We prove $d_{\OY{}}(u,v)\leq \tau_k|uv|$ by induction on the length $|uv|$ for $u,v\in\pset$.
    Notice that it is always possible to find $j$ so that $v\in C_u\left(2j\pi/k+\pi/4,2(j+1)\pi/k+\pi/4\right)$.
    See Figure~\ref{fig:oy} for an illustration.
    This cone is contained in $C_u(\gamma_1,\gamma_2)$,
    where $\gamma_1=2j\pi/k$ and $\gamma_2=2j\pi/k+\gamma$, since
    $k>24$ and thus $\gamma\geq \pi/2$.
    If $\ov{uv}$ is the shortest in $C_u(\gamma_1,\gamma_2)$,
    then $\ov{uv}\in \OY{k}$ and
    we are done with the proof.

    Now, suppose that in the construction of construction of $\OY{k}$,
    we choose $\ov{uw}$ for $C_u(\gamma_1,\gamma_2)$ with $w\neq v$. Since $k>24$, one shall see that
    \begin{eqnarray*}
        \alpha &: =&\max\{\varphi(uv)-\gamma_1,\gamma_2-\varphi(uv)\} \\
                &\leq &\max\{2\pi/k+\pi/4,\gamma-\pi/4\}=2\pi/k+\pi/4<\pi/3.
    \end{eqnarray*}
    Hence, we can apply Corollary \ref{col:ineq} on the sector $R=\overline{D(u,|uv|)\cap C_u(\gamma_1,\gamma_2)}$,
    which claims
    $$
    \dfrac{|uw|}{|uv|-|vw|}\leq \tau_k=\left(1-2\sin \frac{\alpha}{2}\right)^{-1}.
    $$
    Since $\alpha<\pi/3$, we have $|vw|<|uv|$.
    By the induction hypothesis, we can see that
    $$
    d_{\OY{}}(u,v)\leq |uw|+d_{\OY{}}(v,w)\leq |uw|+\tau_k|vw|
    \leq \tau_k(|uv|-|vw|)+\tau_k|vw|=\tau_k|uv|,
    $$
    which completes the proof.
  \end{proof}

 We also note that Barba et al. proposed the so-called \textit{continuous Yao graphs} in \cite{barba2014continuous},
 which play a similar role as $\OY{k}$.
 By adapting their method, it might be possible to prove a slightly smaller stretch factor for $\OY{k}$.
 However the constant is not our primary goal in this paper, and we leave it to the future work.

 Now we define the trapezoidal Yao graphs based on the curved trapezoid $\ttheta$
 (Definition~\ref{def:ct}).
  \begin{defi}[Trapezoidal Yao Graph $\TY{k}$] \label{def:ty}
    Let $\theta=\left\lceil\dfrac{k}{8}\right\rceil\dfrac{2\pi}{k}$. For every $u\in\pset$ and $j=0,\ldots,k-1$,
    define two curved trapezoids
    $$
    \Gamma_{j1}=\rotation{(\ttheta)}{2j\pi/k}
    \qquad\text{and}\qquad
    \Gamma_{j2}=\rotation{(\ttheta^-)}{2j\pi/k}.
    $$
    Note that their bottom sides lie on the ray of angle $2j\pi/k$.
	For each $i=1,2$, we do the following:
	We grow a curved trapezoid $u+\lambda\Gamma_{ji}$ by gradually
	increasing $\lambda$ (initially 0) until its boundary hits some point $v$.
    If the critical arc hits $v$, we select $\ov{uv}$.
    Otherwise, we select nothing.
    All the selected edges form the graph $\TY{k}(\pset)$.
  \end{defi}

    \begin{figure}[t]
    \centering
    \includegraphics[width=0.4\linewidth]{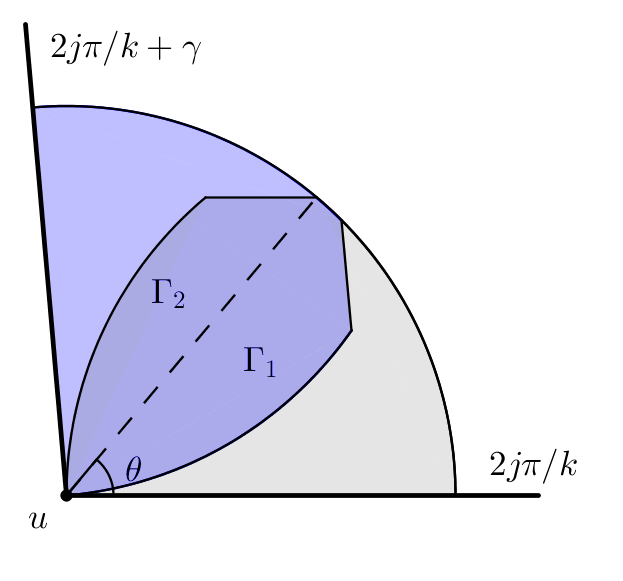}
    \caption{Illustration for Lemma \ref{lm:ty}:
    	Two curved trapezoids cover a sector.}
    \label{fig:ty}
    \end{figure}

  \begin{lm}\label{lm:ty}
    For any integer $k>24$, $\TY{k}$ is a $\tau_k$-spanner
     where $\tau_k=\left(1-2\sin\big(\dfrac{\pi}{k}+\dfrac{\pi}{8}\big)\right)^{-1}$.
  \end{lm}
  \begin{proof}

	Note that in Definition~\ref{def:ct}, we require that $\theta\in[\pi/4,\pi/3)$.
	When $k>24$, we can see that
	$$
	\dfrac{\pi}{4}\leq \theta=\left\lceil\dfrac{k}{8}\right\rceil\dfrac{2\pi}{k}\leq \dfrac{k+7}{k}\cdot\dfrac{\pi}{4}<\dfrac{\pi}{3}.
	$$
	Hence, the value $\theta$ in Definition~\ref{def:ty} satisfies the requirement of $\theta$ in Definition~\ref{def:ct}.

	We show that the overlapping Yao graph
	$\OY{k}$ is actually a subgraph of $\TY{k}$, from which
	the lemma is an immediate consequence.
    Given any $j$ and one cone $C_u(2j\pi/k,2j\pi/k+\gamma)$, we choose two copies of
    $\ttheta$: $\Gamma_{1}=\rotation{(\ttheta)}{2j\pi/k}$,
    and $\Gamma_{2}=\rotation{(\ttheta^-)}{2j\pi/k+\gamma}$. See Figure \ref{fig:ty}.
    Since $2\theta\geq\gamma\geq\pi/2$,
    it is clear that $\tone$ and $\ttwo$ exactly cover the interior of sector $D(o,1)\cap C_o(2j\pi/k,2j\pi/k+\gamma)$.
    Now, for any edge $\ov{uv}\in \OY{k}$ selected within this cone, both $u+|uv|\tone$ and $u+|uv|\ttwo$ must have empty interior
    ($|uv|$ is the shortest within the cone).
    Therefore, if $\varphi(uv)-2j\pi/k\leq\theta$, it should be selected into $\TY{k}$ with respect to $\tone$.
    Otherwise it should be selected with respect to $\ttwo$.
    This implies that $\OY{k}$ is a subgraph of $\TY{k}$,
    and by Lemma \ref{lm:oy}, $\TY{k}$ is a $\tau_k$-spanner as well.
  \end{proof}

\section{Yao-Yao Graphs $\YY{2k}$ are Spanners}\label{s:main}

  In this section, we prove our main result that $\YY{2k}$ has a constant stretch factor for large $k$, by show that $\YY{2k}$ spans $\TY{2k}$.

  To begin with, we prove an important property of $\TY{2k}$, which will be useful
  later.
  The property can be seen as
  an analogue of the property of $\Theta_6$
  shown in \cite[Lemma 2]{bauer2013infinite}.
  One of the reasons we can improve $\YY{6k}$ to more general $\YY{2k}$ is that we can take advantage
  of the curved trapezoid rather than the regular triangle.
  We state this property by assuming, that without loss of generality,
  a curved trapezoid $\ttheta$ is placed at its normal position,
  i.e., $\ttheta$ lies in the first quadrant with two vertices $o(0,0)$ and $p(1,0)$.
  From now on,
  we work with trapezoidal Yao graphs $\TY{2k}$.
  Hence, the associated parameter $\theta$ (see Definition~\ref{def:ty}) is  $\theta=\left\lceil\dfrac{2k}{8}\right\rceil\dfrac{2\pi}{2k}=
  \left\lceil\dfrac{k}{4}\right\rceil\dfrac{\pi}{k}$.

  \begin{lm}\label{lm:ty2}
     Suppose $o\in\pset$ and $\ttheta$ has an empty interior.
     If there is a point $a\in\pset$ such that $0<x_a<1$, $y_a\leq0$ and $0<\varphi(ap)<\pi/6$, then
     \[
     d_{\TY{}}(oa)\leq x_a+(2\tau_{2k}+1)|y_a|
     \]
     where $\tau_k$ is as defined in Lemma~\ref{lm:ty}, and there is a path in $\TY{2k}$ from $a$ to $o$ where each edge is shorter than $|oa|$.
  \end{lm}
  \begin{proof}
    Note that simply applying Lemma \ref{lm:ty} is insufficient to achieve
    the guarantee.
    Instead, we present an iterative algorithm for finding a path from $a$ to $o$.
    The path found by the iterative algorithm is not necessarily a shortest path from $a$ to $o$.
    Nevertheless, we can bound its length as in the lemma.

    \begin{figure}[t]
      \centering
      \includegraphics[width=0.8\linewidth]{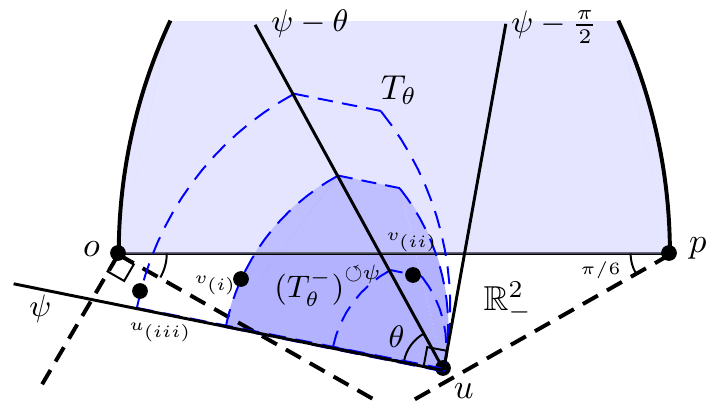}
      \caption{Illustration of three cases in the proof of Lemma \ref{lm:ty2}.}
      \label{fig:typ}
    \end{figure}
    \begin{center}
      \begin{minipage}[h]{\linewidth}
      \begin{tabbing}
        \hspace*{0.8cm}\=\hspace*{0.8cm}\= \kill
        Let $u$ be the current point in $\pset$; Initially $u$ is set to be $a$. \\
        While $u\neq o$ and $\varphi(ou)>-\dfrac{\pi}{6}$ do \\
        \> Let $\psi=\min\{j\pi/k\mid j\pi/k>\varphi(uo), j=0,1,\ldots,2k-1\}$; \\
        \> Grow a curved trapezoid $u+\lambda\rotation{(\ttheta^-)}{\psi}$ until its boundary hits some point $v$. \\
        \> If $v\in C_u(\psi-\dfrac{\pi}{2},\psi-\theta)$: \\
         \> \> (i) $\TY{2k}$ must contain the edge $\ov{uv}$. We add $\ov{uv}$ to our path; \\
         \> Otherwise: \\
         \> \> (ii) Add the path from $u$ to $v$ in $\OY{2k}$ constructed in Lemma~\ref{lm:oy}; \\
        \> Set the current point $u$ to be $v$ and proceed to the next iteration. \\
        If now $u=o$, the path is already found; \\
        (iii) Otherwise take the path from $u$ to $o$ in $\OY{2k}$ constructed in Lemma~\ref{lm:oy}.
      \end{tabbing}
      \end{minipage}
    \end{center}

    See Figure~\ref{fig:typ} for an illustration of the three cases (i),(ii) and (iii).
    Let $l(u)$ be the length of the path from $u$ to $o$ generated by the algorithm above.
    To analyze its behaviour, we define a potential function
     $$
     \Phi(u)=x_u+(2\tau_{2k}+1)|y_u|-l(u).
     $$
    We claim that the potential function never increases
    as the algorithm proceeds and is eventually $0$.
    The potential only changes in the three labeled steps (i),(ii) and (iii).
    When executing step (i) and (ii), it is clear that $\varphi(uo)\in\left(\dfrac{5}{6}\pi,\pi\right]$, $\varphi(up)\in\left[0,\dfrac{1}{6}\pi\right)$. Therefore $\psi$ also falls in the range $\left(\dfrac{5}{6}\pi,\pi\right]$. Moreover, $v$ must be contained in $\mathbb{R}^2_-$ due to Lemma \ref{lm:lmt}, that is, $y_v\leq 0$.
  \begin{itemize}
    \item[-] In step (i), $v$ is simply the nearest neighbor in the cone $C_u(\psi-\theta,\psi)$.
    Since $\varphi(up)<\pi/6$, we know $|uv|<|vp|$. Since $\psi-\theta>\pi/2$ and $\psi\leq\pi$,
    we can see that $x_v\leq x_u$ and $|y_v|<|y_u|$. Hence the change in potential is
    \begin{align*}
    \Delta\Phi=\Phi(v)-\Phi(u) &=|uv|-(x_u-x_v)-(2\tau_{2k}+1)(|y_u|-|y_v|) \\
    &\leq|uv|-(x_u-x_v)-(|y_u|-|y_v|)\leq 0.
    \end{align*}

    \item[-] In step (ii), since $\varphi(uv)\in[\psi-\dfrac{\pi}{2},\psi-\theta)$, we know $|\varphi(uv)-\pi/2|$ is at most $\pi/4$, so $|y_v|<|y_u|$ and $|x_v-x_u|<|y_u|-|y_v|$. That further implies $|uv|<2(|y_u|-|y_v|)$, and thus $d_{\TY{}}(uv)\leq\tau_{2k}|uv|<2\tau_{2k}(|y_u|-|y_v|)$. Therefore the change in potential is
    \begin{align*}
      \Delta\Phi=\Phi(v)-\Phi(u) &=d_{\TY{}}(uv)-(x_u-x_v)-(2\tau_{2k}+1)(|y_u|-|y_v|) \\
      &=(d_{\TY{}}(uv)-2\tau_{2k}(|y_u|-|y_v|))+(x_v-x_u-(|y_u|-|y_v|))\leq 0.
    \end{align*}

    \item[-] In step (iii), suppose the node arrived before $u$ is $u'$, with $\varphi(u'o)>\dfrac{5}{6}\pi$. Since $\ov{u'u}$ is chosen, it must be the case that $|u'u|\leq|u'o|$. It is immediate that $\varphi(uo)\in\left(\dfrac{\pi}{3},\dfrac{5}{6}\pi\right]$, and thus $|uo|\leq 2|y_u|$, $-x_u\leq |y_u|$. Therefore, we have that
    \[
    \Delta\Phi=\Phi(o)-\Phi(u)=(d_{\TY{}}(uo)-2\tau_{2k}|y_u|)+(-x_u-|y_u|)\leq0.
    \]
  \end{itemize}
  Now that $\Phi(o)=0$, and the above arguments show that the potential cannot increase during the path from $u$ to $o$, we can conclude that $\Phi(a)\geq 0$. That is, $l(a)\leq x_a+(2\tau_{2k}+1)|y_a|$.

  Also notice that in steps (i) and (ii), it is ensured that $|uv|<|ou|$ and $|ov|<|ou|$ since $v$ is contained in the curved trapezoid. On the other hand, the paths from $\OY{2k}$ in steps (ii) and (iii), which are constructed according to Lemma~\ref{lm:oy}, contains only edges shorter than the direct distance. Thus the edges we selected are all shorter than $|oa|$.
  \end{proof}

  \begin{lm}\label{lm:main}
    On a Yao-Yao graph $\YY{2k}(\pset)$ with $k\geq 42$, if $\ov{uv}\in \TY{2k}$, then $d_{\YY{}}(uv)\leq \tau_k'|uv|$,
    where $\tau_k'=\sqrt{2}+O(k^{-1})$ is a constant depending only on $k$.
  \end{lm}

  \begin{proof}
    Let $\tau_k'\geq1$ be a constant to be fixed later.
    Our proof is by induction on the length $|uv|$.
    The base case is simple:
    the shortest edge is in both $\TY{2k}$ and $\YY{2k}$.
	Now, consider an edge $\ov{uv}\in \TY{2k}$.
    Without loss of generality, we assume that
    $|uv|=1$. We further assume that $\ttheta$, which generates $\ov{uv}$ in $\TY{2k}$, in its normalized position, i.e.,
    $u$ is $o(0,0)$ and another vertex is $p(1,0)$.
    See Figure~\ref{fig:main}.
    Noticing that $\theta$ is a multiple of $\pi/k$, we can see that
    $v$ must be the nearest neighbor of $u$ in its corresponding cone of $\YY{2k}$.

    If $\ov{uv}\in \YY{2k}$, then $d_{\YY{}}(uv)=|uv|$ and we are done with the proof.
    Otherwise, there must be another edge $\ov{wv}\in \YY{2k}$ where $u$ and $w$ are in the same $v$-apex cone,
    and $|vw|\leq|uv|$ (this is due to the construction of Yao-Yao graph).
    It follows that $|uw|<|uv|$ since $k\geq 42$.

    Now, we prove by induction that $d_{\YY{}}(uv)\leq \tau_k'|uv|$.
    Now noticing that $\ttheta$ has an empty interior, there are only two cases for $w$ we need to consider.
    See Figure \ref{fig:main} for an example.
    \begin{enumerate}
      \item $y_w>0$: See the left part of Figure~\ref{fig:main}.
      By Lemma \ref{lm:ineq}, we can see that within all possible positions of $w$, $|vw|/(|uv|-\tau_{2k}|uw|)$ can only be maximized when $|vw|=|uv|$ or when $w$ is on the left boundary of $\ttheta$. Since when $k\geq 42$, it holds $2\tau_{2k}\sin\frac{\pi}{2k}<1$, therefore
      $$
      \dfrac{|vw|}{|uv|-\tau_{2k}|uw|}\leq\dfrac{1}{1-2\tau_{2k}\sin\frac{\pi}{2k}}.
      $$
      Hence, if $\tau_k'\geq 1/(1-2\tau_{2k}\sin\frac{\pi}{2k})$, we have that
      $
      |vw|\leq\tau_k'(|uv|-\tau_{2k}|uw|).
      $
      Consider the shortest path from $u$ to $w$ in $\TY{2k}$.
      Notice that $d_{\TY{}}(uw)\leq \tau_{2k}|uw|<|uv|$.
      Applying the induction hypothesis
      to every edge of the shortest path,
      we can see $d_{\YY{}}(uw)\leq \tau_k'd_{\TY{}}(uw)$.
      Combining the above inequalities, we can finally obtain that
      \[
      d_{\YY{}}(uv)\leq d_{\YY{}}(uw)+|vw|\leq \tau_k'(d_{\TY{}}(uw)+|uv|-\tau_{2k}|uw|)\leq \tau_k'|uv|.
      \]

      \item $y_w\leq 0$: See the middle and right parts of Figure~\ref{fig:main}.
      Consider the two cases of angle $\varphi(uv)$:
      \begin{itemize}
        \item If $0\leq\varphi(uv)\leq \pi/k$, then $\varphi(wp)\leq\varphi(wv)\leq\pi/k$;
        \item If $\pi/k<\varphi(uv)\leq\theta$, then $\angle upv\leq \frac{(k-1)\pi}{2k}\leq\angle uwv$. Thus $p$ is outside the circumcircle of $\triangle uvw$, and $\varphi(wp)=\angle upw\leq\angle uvw\leq\pi/k$.
      \end{itemize}
      Therefore in both cases, we have $0<x_w<x_p$ and $0\leq \varphi(wp)\leq\pi/k$. We notice the following two inequalities:
	  \begin{align*}
      &|y_w|\leq(x_p-x_w)\tan\varphi(wp)\leq(|uv|-x_w)\tan\frac{\pi}{k}, \\
      &|vw|\leq(x_v-x_w)\sec\varphi(wv)\leq(|uv|-x_w)\sec\left(\theta+\frac{\pi}{k}\right).
      \end{align*}
      Consider the path from $w$ to $u$ as we constructed in Lemma~\ref{lm:ty2}
      ($w$ and $u$ correspond to $a$ and $o$ in the lemma, respectively).
      Since
      all edges in the constructed path are shorter than $|wu|$ (thus shorter than
      $|uv|$ as well), we can apply the induction hypothesis to get
      $d_{\YY{}}(uw)\leq \tau_k'd_{\TY{}}(uw)$.
      Applying Lemma \ref{lm:ty2}, we have
      \begin{align*}
      d_{\YY{}}(uv) & \leq d_{\YY{}}(uw)+|vw|\leq \tau_k'd_{\TY{}}(uw)+|vw|  \\
      &\leq \tau_k'(x_w+(2\tau_{2k}+1)|y_w|)+|vw| \\
      &\leq \tau_k'x_w +
      \left(\tau_k' (2\tau_{2k}+1)\tan\frac{\pi}{k}+\sec\left(\theta+\frac{\pi}{k}\right)\right)
      (|uv|-x_w).
      \end{align*}
      Suppose we choose $\tau_k'$ such that
      $
      \tau_k'(2\tau_{2k}+1)\tan\frac{\pi}{k}+\sec(\theta+\frac{\pi}{k})\leq\tau_k',
      $
      (which is equivalent to
	  $
      \tau'\geq ((1-(2\tau_{2k}+1)\tan\frac{\pi}{k})\cos(\theta+\frac{\pi}{k}))^{-1}
      $
      ).
      Note that the above holds also because
      $(2\tau_{2k}+1)\tan\frac{\pi}{k}<1$ when $k\geq 42$.
      Finally, we can conclude that
      $$
      d_{\YY{}}(uv) \leq \tau'_k |uv|.
      $$
    \end{enumerate}
    \begin{figure}[t]
      \centering
      \includegraphics[width=0.95\linewidth]{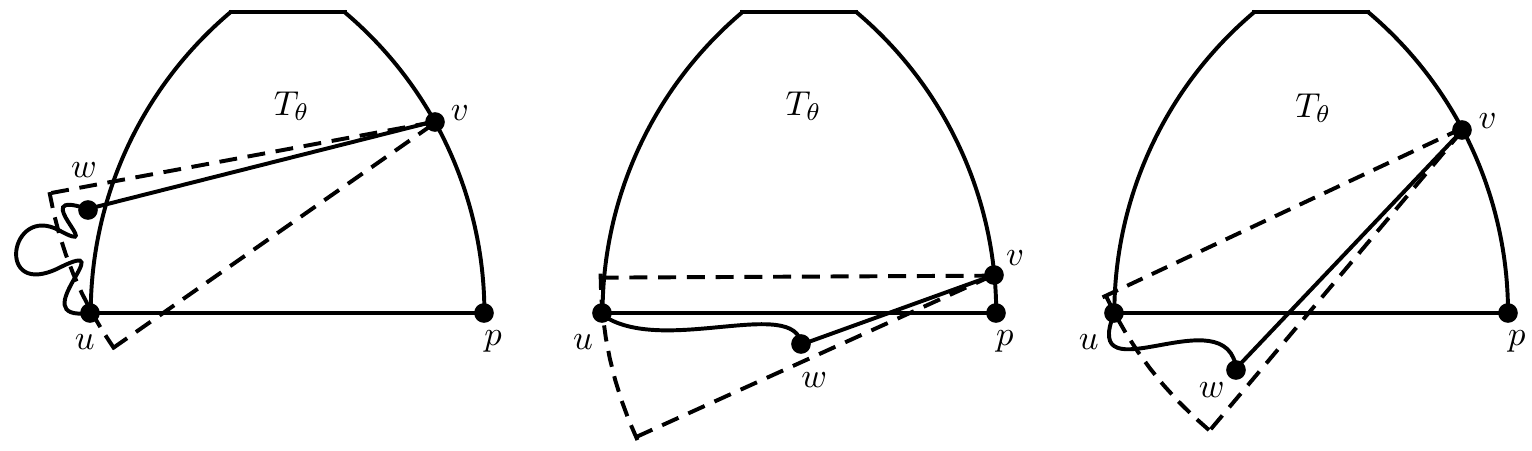}
      \caption{Proof for Lemma \ref{lm:main}. Left: $w$ is in the upper half-plane; Middle: $w$ is in the lower half-plane and $\varphi(uv)\leq\pi/k$; Right: $w$ is in the lower half-plane and $\varphi(uv)>\pi/k$.}
      \label{fig:main}
    \end{figure}
    We can conclude from the above two cases
    that $\tau_k'$ can be chosen to be $\sec\frac{\pi}{4}+O(k^{-1})=\sqrt{2}+O(k^{-1})$.
  \end{proof}
  Combining Lemma \ref{lm:ty} and \ref{lm:main},
  with the observation that $\tau_{k}=(1-2\sin(\pi/8))^{-1}+O(k^{-1})$,
  we can get our main result that almost all $\YY{2k}$ are spanners.
  \setcounter{thm}{1}
  \begin{thm} (restated)
    For any integer $k\geq42$, $\YY{2k}$ is a $t_k$-spanner, where $t_k=\tau_k'\tau_{2k}=\sqrt{2}(1-2\sin(\pi/8))^{-1}+O(k^{-1})=6.03+O(k^{-1})$.
  \end{thm}

  \section{Conclusion and Future Work}\label{s:con}

  In this paper we proved that Yao-Yao graphs $\YY{2k}$ are geometric spanners for $k$ large enough,
  making a positive progress to the long-standing Conjecture~\ref{cj}.
  For Yao-Yao graphs with odd parameters, the resolution of the conjecture is still elusive and
  seems to require additional new ideas.

  We did not try very hard to optimize the constant stretch ratio for different $k$ values.
  Hence, the constant claimed in Theorem \ref{thm:main} may well be further improved.
  One potential approach is to use techniques developed in \cite{barba2014continuous}.
  Obtaining tighter bounds (and lower bounds as well!) is left as an interesting future work.

  We propose some other potential future directions, which are interesting in their own rights
  and might lead to sparse spanners with bounded degrees and small stretch factors.
  \begin{itemize}
    \item
        Consider the variant of Yao-Yao graphs where
        the cones are not divided uniformly.
        One could expect that such variants are spanners if the apex angles of the cones are all small.
        It seems that our techniques in the paper can be adapted to this situation if the separated cones
        are \textit{centrosymmetric}: every separation ray of polar angle $\varphi$ has a corresponding one at polar angle $\pi+\varphi$.
        But the more general case is still open.
    \item Instead of insisting on one polar division for every node,
        we can let each point choose its polar axis individually by uniformly and independently randomizing a direction.
        We conjecture that
        such graphs are also spanners (in expectation or with high probability) for $k$ large enough.
  \end{itemize}

  \bibliography{p66-Li}

\end{document}